\documentclass[preprint]{elsarticle}

\usepackage{lineno,hyperref}
\modulolinenumbers[5]

\journal{Journal of \LaTeX\ Templates}
\usepackage{algorithm}
\usepackage{amsmath}
\usepackage{bm}
\usepackage{graphicx}
\usepackage{amssymb}
\usepackage{epstopdf}
\usepackage{array,longtable,multirow}
\usepackage{multicol}
\usepackage{enumitem}
\usepackage{pdflscape}
\usepackage{graphicx}
\usepackage{amsmath,amssymb,amsfonts,epsfig,amsthm}
\usepackage{algpseudocode}
\usepackage{color}
\usepackage[usenames,dvipsnames]{xcolor}
\usepackage{fancyhdr}
\usepackage{verbatim}
\usepackage{etoolbox}
\usepackage{url}
\usepackage{enumitem}
\usepackage{tikz-cd}
\usepackage{tikz}
\usepackage[utf8]{inputenc}
\usepackage{adjustbox}
\usepackage{cases}

\makeatletter
\renewcommand*\env@matrix[1][*\c@MaxMatrixCols c]{%
  \hskip -\arraycolsep
  \let\@ifnextchar\new@ifnextchar
  \array{#1}}
\makeatother

\newtheorem{theorem}{Theorem}[section]

\newtheorem{lemma}[theorem]{Lemma}

\theoremstyle{definition}
\newtheorem{definition}[theorem]{Definition}
\newtheorem{remark}{Remark}
\newtheorem*{notation}{Notation}

\newtheorem{example}{Example}

\newcommand{\R}{\mathfrak R[x]/\langle x^n-1\rangle}
\newcommand{\FR}{\mathbb F_{p^m}[x]/\langle x^n-1\rangle}
\newcommand{\la}{\langle}
\newcommand{\ra}{\rangle}

\begin{document}

\begin{frontmatter}

\title{Determination for minimum symbol-pair and RT weights via torsional degrees of repeated-root cyclic codes}

\author{Boran Kim}
\address{Department of Mathematics Education, Teachers College, Kyungpook National University, 80 Daehakro, Bukgu, Daegu 41566 , Republic of Korea.}
\fntext[myfootnote]{E-mail address: bkim21@knu.ac.kr (Boran Kim)\\
Boran Kim is supported by Basic Science Research Program through the National Research Foundation of Korea(NRF) grant funded by the Korea government(MSIT)(NRF-2021R1C1C2012517).}


\begin{abstract}
There are various metrics for researching error-correcting codes.
Especially, high-density data storage system gives the existence of the inconsistency for the reading and writing process.
The symbol-pair metric is motivated for outputs that have overlapping pairs of symbols in a certain channel.
The Rosenbloom-Tsfasman (RT) metric is introduced since there exists a problem that is related to transmission over several parallel communication channels with some channels not available for the transmission.
In this paper, we determine the minimum symbol-pair weight and RT weight of repeated-root cyclic codes over $\mathfrak R=\Bbb F_{p^m}[u]/\la u^4\ra$ of length $n=p^k$.
For the determination, we explicitly present third torsional degree for all different types of cyclic codes over $\mathfrak R$ of length $n$.

\end{abstract}


\end{frontmatter}

\section{Introduction}
In coding theory, one of the interesting goals is for constructing codes that give easier encoding and decoding processes. 
By the perspective, cyclic codes have importance in this theory because they have efficient encoding and decoding algorithms.
Moreover, the cyclic codes give many important applications in the other areas, such as cryptography and sequence design \cite{CDY, DW, DYT}.
In particular, the repeated-root cyclic codes are first introduced in \cite{CDL, SR}.
They present that repeated-root cyclic codes have a concatenated construction and many optimal codes; from this, many researchers have been studying repeated-root cyclic codes over various rings.

There are various metrics for researching error-correcting codes; such as Hamming, Lee, Euclidean, symbol-pair, and RT metrics.
Especially, high-density data storage system gives the existence of the inconsistency for the reading and writing process.
The symbol-pair metric is motivated for outputs that have overlapping pairs of symbols in a certain channel \cite{CB, CB2}.
The Rosenbloom-Tsfasman (RT) metric is introduced \cite{RT} since there exists a problem that is related to transmission over several parallel communication channels with some channels not available for the transmission.

For many years, linear codes have active developments over various rings \cite{GD, HKKL, KDS, KL1, KL2, KL3, LY, S, HKN}.
Recently, finding symbol-pair weight and RT weight is a great challenge in coding theory \cite{CLL, LG, SZW}.
In \cite{DKKSSY}, they study MDS symbol-pair repeated-root constacyclic codes over $\mathbb F_{p^m}+u\mathbb F_{p^m}$.
The author gives explicit classification for cyclic codes over $\mathbb F_{p^m}[u]/\la u^3 \ra$ with torsional degrees \cite{KLD}. 
Sharma and Sidana \cite{SS} consider symbol-pair weight and RT weight for repeated-root constacyclic codes over finite commutative chain rings.
They give the relationship between a certain minimum weight and torsional degrees for constacyclic codes.
In detail, if we know about the $i$-th torsional degree for constacyclic codes over $\mathbb F_{p^m}[u]/\la u^{i+1}\ra$, then we can determine the symbol-pair weight and RT weight.
Our results are motivated by this fact.
The determination for $i$-th torsional degree is very difficult in general. 
So we focus on the specific ring $\mathbb F_{p^m}[u]/\la u^4 \ra$, and we explicitly determine the symbol-pair weight and RT weight by using third torsional degree. 

In this paper, we determine the third torsional degrees for all different types of repeated-root cyclic codes over $\mathfrak R=\mathbb F_{p^m}[u]/\la u^4 \ra$ of length $n=p^k$ (Theorems \ref{09_07_2020_thm_1}, \ref{09_07_2020_thm_2} and \ref{09_09_2020_thm_1}).
Through this information, we give the symbol-pair weight and RT weight for the cyclic codes (Theorem \ref{09_10_2020_cor_1}).
We present some examples for supporting our results (Examples \ref{09_10_2020_ex_2}, \ref{09_10_2020_ex_3} and \ref{09_17_2020_ex_1}).

\section{Preliminaries}\label{pre}
Let $\mathfrak R$ be a finite commutative ring with unity.
An $\mathfrak R$-submodule of $\mathfrak R^n$ is called a {\it code} $C$ of length $n$ over a ring $\mathfrak R$.
Any element $c=(c_1,\ldots,c_n)$ in $C$ is called a {\it codeword}.
Henceforth, a code means a linear code, and $\mathfrak R$ is the ring $\mathbb F_{p^m}[u]/\langle u^4 \rangle$, where $p$ is a prime number and $m\ge 1$; the ring $\Bbb F_{p^m}[u]/\la u^4 \ra$ is a finite commutative Frobenius ring.

\begin{definition}\label{09_09_2020_def_1}
Let $\mathfrak R$ be a finite commutative Frobenius ring, and $g(x)$ be a polynomial in $\mathfrak R[x]$.
A {\it polycyclic code} over $\mathfrak R$ is an ideal in $\mathfrak R[x]/\la g(x)\ra$.
\begin{itemize}
    \item[{(i)}] If $g(x)=x^n-1$, then the code is called a {\it cyclic code} of length $n$.
    \item[{(ii)}] If $g(x)=x^n+1$, then the code is called a {\it negacyclic code} of length $n$.
    \item[{(iii)}] If $g(x)=x^n+\lambda$, then the code is called a {\it constacyclic code} of length $n$, where $\lambda$ is a unit in $\mathfrak R$.
\end{itemize}
\end{definition}

In this paper, we deal with a repeated-root cyclic code over $\mathfrak R=\Bbb F_{p^m}[u]/\la u^4 \ra$ of length $n=p^k$.
We present a cyclic code $C$ over $\mathfrak R$ of length $n$ as an ideal in $\R$ by Definition \ref{09_09_2020_def_1}.
We consider a codeword in $C$ as a polynomial in $\R$.
Furthermore, the generator polynomials for a cyclic code over $\mathfrak R$ are given in Lemma \ref{09_09_2020_lem_1}.
Before we give the generator polynomials, we introduce a torsion code of a cyclic code over $\mathfrak R$ of length $n$.

\begin{definition}
Let $C$ be a cyclic code over $\mathfrak R$ of length $n=p^k$. 
Then the {\it $i$-th torsion code} of $C$ is defined as
$$Tor_i(C)=\{\mu(g(x))\in \FR: u^i g(x)\in C\},$$
where $\mu$ is a natural projection map from $\R$ to $\FR$ $(0 \le i \le 3)$.
Especially, when $i=0$, the code $Tor_0(C)$ is called the {\it residue code} of $C$. 
\end{definition}

In \cite{SS}, the authors give the following results, Lemmas \ref{09_09_2020_lem_2} and \ref{09_09_2020_lem_1}, for the ring $\Bbb F_{p^m}[u]/\la u^{i+1} \ra$.
In this current work, we deal with the ring $\mathfrak R=\Bbb F_{p^m}[u]/\la u^4 \ra$, thus the lemmas are rewritten for this ring $\mathfrak R$.

\begin{lemma}\cite[Theorem 9]{SS}\label{09_09_2020_lem_2} 
Let $C$ be a cyclic code over $\mathfrak R$ of length $n=p^k$.
\begin{itemize}
    \item[{\em (i)}] The $i$-th torsion code $Tor_i(C)$ of $C$ is a cyclic code over $\mathbb F_{p^m}$, and $$Tor_i(C)=\langle (x-1)^{t_i}\rangle,$$ 
    where $t_i$ is an integer satisfying $0 \le t_i \le p^k$.
    The integer $t_i$ is called the {\it $i$-th torsional degree} of $C$ $(0 \le i \le 3)$.
    \item[{\em (ii)}] We have $0 \le t_3 \le t_2 \le t_1 \le t_0 \le p^k$.
\end{itemize}
\end{lemma}

\begin{lemma}\cite[Theorem 10]{SS}\label{09_09_2020_lem_1}
Let $C$ be a cyclic code over $\mathfrak R$ of length $n=p^k$.
Let $Tor_i(C)$ be the $i$-th torsional code of $C$ with the $i$-th torsional degree $t_i$ $(0 \le i \le 3)$.
Then the code $C$ is uniquely generated by the following four polynomials in $\R$:
\[
\begin{array}{l}
g_0(x)=(x-1)^{t_0}+u(x-1)^{k_1}p_1(x)+u^2(x-1)^{k_2}p_2(x)+u^3(x-1)^{k_3}p_3(x),\\
g_1(x)=u(x-1)^{t_1}+u^2(x-1)^{k_4}p_4(x)+u^3(x-1)^{k_5}p_5(x),\\
g_2(x)=u^2(x-1)^{t_2}+u^3(x-1)^{k_6}p_6(x),\\
g_3(x)=u^3(x-1)^{t_3},
\end{array}
\]
where $k_1<t_1$, $k_i<t_2$, $k_j<t_3$, and
$p_\ell(x)$ is a unit or zero in $\FR$ $(i=2,4, j=3,5,6 \mbox{~and~}\ell=1,\ldots,6)$.
Especially, if $t_i=p^k$ (resp. $t_i=0$), then $g_i(x)=0$ (resp. $g_i(x)=u^i$) for all $0 \le i \le 3$.
\end{lemma}

Let $g(x)=(x-1)^{s_0}h_0(x)+u(x-1)^{s_1}h_1(x)+u^2(x-1)^{s_2}h_2(x)+u^3(x-1)^{s_3}h_3(x)$ be an arbitrary polynomial in $\R$, where $h_0(x)=0$ or $1$, $h_i(x)$ is a unit or zero in $\FR$ and $s_i\ge 0$ $(1 \le i \le 3)$.
In the polynomial $g(x)$, the term $u(x-1)^{s_1}h_1(x)$ (resp. $u^2(x-1)^{s_2}h_2(x)$, $u^3(x-1)^{s_3}h_3(x)$) of $g(x)$ is called a {\it $u$-part} (resp. {\it $u^2$-part, $u^3$-part}) of $g(x)$.

\smallskip

We note that, for any codeword $f(x)=\sum_{i=0}^{n-1}c_i(x-1)^i$ in $C$, we also consider the codeword $f(x)$ as the vector ${\bf c}=(c_0,\ldots,c_{n-1})$.

The next definition gives various metrics for the weight of a codeword in a code.

\begin{definition}
Let $C$ be a code over $\mathfrak R$ of length $n$, and ${\bf c} =(c_0,\ldots,c_{n-1})$ be a codeword in $C$.
\begin{itemize}
    \item[{(i)}] The {\it Hamming weight} of ${\bf c}$ is the number of $i$ such that $c_i\neq 0$ in ${\bf c}$ for $0\le i \le n-1$.
    \item[{(ii)}] The {\it symbol-pair weight} of ${\bf c}$ is the number of $i$ such that $(c_i,c_{i+1})\neq (0,0)$ in $\pi({\bf c})$ for $0\le i \le n-1$, where $\pi({\bf c}):=((c_1,c_2),(c_2,c_3),\ldots,(c_{n-1},c_0))$ is the symbol-pair vector of the codeword ${\bf c}$.
    \item[{(iii)}] The {\it Rosenbloom-Tsfasman (RT) weight} of ${\bf c}$ is $1+\max\{i:0\le i \le n-1, c_i\neq 0\}$ in ${\bf c}$ when ${\bf c}\neq {\bf 0}$. If ${\bf c}={\bf 0}$, then the RT weight is equal to $0$.
\end{itemize}

\end{definition}

The minimum weight is the minimum value among the weights for all nonzero codewords in a code $C$.
The {\it minimum Hamming weight}, {\it minimum symbol-pair weight} and {\it minimum RT weight} of a code $C$ are denoted by $wt_H(C)$, $wt_{sp}(C)$ and $wt_{RT}(C)$, respectively.

\section{Third torsional degree for a cyclic code over $\mathfrak R$ of length $n=p^k$}\label{sec3}

In this section, we determine the third torsional degree for a cyclic code over $\mathfrak R$ of length $n=p^k$. 
After that, we present the minimum symbol-pair weight and RT weight for the code.
From this time on, we use the following notations.

\begin{notation}
\[
\begin{array}{l}
\bullet~ \mathfrak R=\mathbb F_{p^m}[u]/\langle u^4 \rangle\\
\bullet~ n=p^k \mbox{,~where $p$ is prime and~} k\ge 1.\\
\bullet~ \binom{s_1}{s_2} \mbox{~the binomial coefficient $s_1$ choose $s_2$}\\
\bullet~ g_0(x)=(x-1)^r + u(x-1)^{k_1}p_1(x) + u^2(x-1)^{k_2}p_2(x) + u^3(x-1)^{k_3}p_3(x)\\
\bullet~ g_1(x)=u(x-1)^{r_1} + u^2(x-1)^{k_4}p_4(x) + u^3(x-1)^{k_5}p_5(x)\\
\bullet~ g_2(x)=u^2(x-1)^{r_2} + u^3(x-1)^{k_6}p_6(x)\\
\bullet~ g_3(x)=u^3(x-1)^{r_3} \\
\mbox{where~} 0\le r_3\le r_2 \le r_1 \le r < n,
k_1<r_1, k_i<r_2, k_j<r_3, \mbox{~and~} p_\ell(x) \mbox{~is a unit}\\
\mbox{or zero in~} \R
~(i=2,4, j=3,5,6 \mbox{~and~}\ell=1,\ldots,6).
\end{array}
\]
\end{notation}

\begin{remark}
Over the ring $\mathfrak R=\mathbb F_{p^m}[u]/\la u^4 \ra$, all different types of cyclic codes of length $n=p^k$ are given as follows:
\begin{equation*}
\begin{array}l
\bullet~ \mbox{Principal ideals}\\
\la g_0(x) \ra, 
\la g_1(x) \ra, 
\la g_2(x) \ra, 
\la g_3(x) \ra.\\
\bullet~ \mbox{Non-principal ideals}\\
\la g_0(x), g_1(x) \ra, 
\la g_0(x), g_2(x) \ra, 
\la g_0(x), g_3(x) \ra,
\la g_1(x), g_2(x) \ra,
\la g_1(x), g_3(x) \ra,\\
\la g_2(x), g_3(x) \ra,
\la g_0(x), g_1(x), g_2(x) \ra,
\la g_0(x), g_1(x), g_3(x) \ra,
\la g_0(x), g_2(x), g_3(x) \ra,
\\
\la g_1(x), g_2(x), g_3(x) \ra,
\la g_0(x), g_1(x), g_2(x), g_3(x) \ra.

\end{array}    
\end{equation*}
\end{remark}

The third torsional degree for a cyclic code $C$ over $\mathcal R$ is defined in Lemma \ref{09_09_2020_lem_2}.
As the other perspective, the third torsional degree for a cyclic code $C$ over $\mathfrak R$ is the smallest non-negative integer $s$ such that $u^3(x-1)^s \in C$.


\begin{theorem}\label{09_07_2020_thm_1}
\begin{itemize}
\item[{\em (i)}]
For a cyclic code $C=\la g_1(x)\ra$ over $\mathfrak R$ of length $n$, the third torsional degree $t_3$ is given as follows:

{\em (a)} If $n-r_1+k_4\ge r_1$, then the third torsional degree $t_3$ is 
\[
t_3=\begin{cases}
\min\{r_1,\tau_1\} & \mbox{if $\tilde h_1(x)$ is a unit},\\
r_1 & \mbox{if $\tilde h_1(x)=0$},
\end{cases}
\]
where $u^3(x-1)^{\tau_1}\tilde h_1(x)=u^3(x-1)^{n-2r_1+2k_4}p_4^2(x)-u^3(x-1)^{n-r_1+k_5}p_5(x)$,
with a unit or zero $\tilde h_1(x)$ in $\R$.

{\em (b)} If $n-r_1+k_4 < r_1$, then the third torsional degree $t_3$ is 
\[
t_3=\begin{cases}
\min\{n-r_1+k_4,\tau_2\} & \mbox{if $\tilde h_2(x)$ is a unit},\\
n-r_1+k_4 & \mbox{if $\tilde h_2(x)=0$},
\end{cases}
\]
where
$u^3(x-1)^{\tau_2}\tilde h_2(x)=u^3(x-1)^{k_4}p_4^2(x)-u^3(x-1)^{r_1-k_4+k_5}p_5(x)$,
with a unit or zero $\tilde h_2(x)$ in $\R$.

\item[{\em (ii)}]
For a cyclic code $C=\langle g_1(x),g_2(x)\rangle$ over $\mathfrak R$ of length $n$, the third torsional degree $t_3$ is obtained as follows:
We consider\\

$u^3(x-1)^{\tau_3} \tilde h_3(x)$ as
{\small
\begin{equation}\label{09_09_2020_eq_5}
\begin{cases}
u^3(x-1)^{n-r_1+k_5}p_5(x)-u^3(x-1)^{n-r_1-r_2+k_4+k_6}p_4(x)p_6(x) & \mbox{if~} n-r_1+k_4 > r_2,\\
u^3(x-1)^{r_2-k_4+k_5}p_5(x)-u^3(x-1)^{k_6}p_4(x)p_6(x) & \mbox{if~} n-r_1+k_4 \le r_2,\\
\end{cases}
\end{equation}}
where $\tilde h_3(x)$ is a unit or zero in $\R$, and $\tau_3\ge 0$.
Set
\begin{equation}\label{09_09_2020_eq_1}
 u^3(x-1)^{\tau_4} \tilde h_4(x)=u^3(x-1)^{k_4}p_4(x)-u^3(x-1)^{r_1-r_2+k_6}p_6(x),  
\end{equation}
where $\tilde h_4(x)$ is a unit or zero in $\R$, and $\tau_4\ge 0$.
Then, for $j=3,4$,
\[
t_3=\begin{cases}
\min \{\kappa, t, r_2, n-r_1+k_4, n-k_4+k_5, n-r_2+k_6\} & \mbox{if $\tilde h_j(x)$ is a unit}\\
&\mbox{for some $j$},\\
\min \{t, r_2, n-r_1+k_4, n-k_4+k_5, n-r_2+k_6\} & \mbox{if $\tilde h_j(x)=0$}\\
&\mbox{for all $j$},
\end{cases}
\]
where $\kappa=\min\{\tau_j: \tilde h_j(x) \mbox{~is a unit in $\R$ for~}j=3,4\}$, and $t$ is the third torsional degree for the code $\la g_1(x)\ra$ obtained in {\em (i)}.

\item[{\em (iii)}]
For a cyclic code $C=\langle g_2(x)\rangle$ over $\mathfrak R$ of length $n$, the third torsional degree $t_3$ is $$t_3=\min\{n-r_2+k_6, r_2\}.$$

\item[{\em (iv)}]
For a cyclic code $C=\la g_3(x) \ra$ over $\mathfrak R$ of length $n$, then the third torsional degree $t_3$ is $r_3$.
\end{itemize}
\end{theorem}

\begin{proof}
\begin{itemize}
    \item[{(i)}] For $C=\la g_1(x)\ra$, the following two polynomials are all polynomials which have $u^2$-part generated by $g_1(x)$:
    \begin{equation}\label{09_09_2020_eq_2}
    ug_1(x)=u^2(x-1)^{r_1}+u^3(x-1)^{k_4}p_4(x),
    \end{equation}
    and
    \begin{equation}\label{09_09_2020_eq_3}
    (x-1)^{n-r_1}g_1(x)=u^2(x-1)^{n-r_1+k_4}p_4(x)+u^3(x-1)^{n-r_1+k_5}p_5(x).
    \end{equation}
    If $n-r_1+k_4\ge r_1$, then for deleting $u^2$-part in both polynomials \eqref{09_09_2020_eq_2} and \eqref{09_09_2020_eq_3}, we compute
    \[
    \begin{array}{l}
    (x-1)^{n-2r_1+k_4}p_4(x)(u^2(x-1)^{r_1}+u^3(x-1)^{k_4}p_4(x))-(x-1)^{n-r_1}g_1(x),\\
    =u^3(x-1)^{n-2r_1+2k_4}p_4^2(x)-u^3(x-1)^{n-r_1+k_5}p_5(x),\\
    =u^3(x-1)^{\tau_1}\tilde h_1(x),
    \end{array}\]
    where $\tilde h_1(x)$ is a unit or zero in $\R$, and $\tau_1 \ge 0$.
    
    So, if $\tilde h_1(x)$ is a unit, then $t_3=\min\{\tau_1,r_1\}$; this is because $u^2g_1(x)=u^3(x-1)^{r_1}\in C$, $u(x-1)^{n-r_1}g_1(x)=u^3(x-1)^{n-r_1+k_4}\in C$ and $r_1 \le n-r_1+k_4 $ by the assumption.
    If $\tilde h_1(x)=0$, then $t_3=r_1$.
    
    For $n-r_1+k_4<r_1$, we have the result by the same process as above.
    
    \item[(ii)] For $C=\la g_1(x),g_2(x)\ra$, the all polynomials which have $u^2$-part are \eqref{09_09_2020_eq_2}, \eqref{09_09_2020_eq_3} and 
    \begin{equation}\label{09_09_2020_eq_4}
    g_2(x)=u^2(x-1)^{r_2}+u^3(x-1)^{r_6}p_6(x).    
    \end{equation}
    By the same reasoning as (i), using equations \eqref{09_09_2020_eq_2}, \eqref{09_09_2020_eq_3}, we have the third torsional degree $t$ for $\la g_1(x)\ra$.
    From equations \eqref{09_09_2020_eq_3} and \eqref{09_09_2020_eq_4}, we have the equation \eqref{09_09_2020_eq_5}.
    Moreover, through equations  \eqref{09_09_2020_eq_2} and \eqref{09_09_2020_eq_4}, we obtain the equation \eqref{09_09_2020_eq_1}.
    
    We have polynomials $u^3(x-1)^{r_1}$, $u^3(x-1)^{n-r_1+k_4}$, $u^3(x-1)^{r_2}$, $u^3(x-1)^{n-k_4+k_5}$ and $u^3(x-1)^{n-r_2+k_6}$;
    in detail, we multiply $u$ to \eqref{09_09_2020_eq_2}, \eqref{09_09_2020_eq_3} and \eqref{09_09_2020_eq_4}.
    And we consider multiplying $(x-1)^{n-r_1}$ to \eqref{09_09_2020_eq_2}, $(x-1)^{r_1-k_4}$ to \eqref{09_09_2020_eq_3} and $(x-1)^{n-r_2}$ to \eqref{09_09_2020_eq_4}.
    Then we can get the above five polynomials.
    Moreover, the five polynomials are all polynomials such that $u$-part and $u^2$-part are zeros with nonzero $u^3$-part generated by $g_1(x)$ and $g_2(x)$.
    We note that $r_2\le r_1$.
    
    As a result, we get the third torsional degree for $\la g_1(x), g_2(x) \ra$ as in the statement.
    
    \item[(iii), (iv)] Clearly, the result follows.
\end{itemize}

\end{proof}

We give some examples for Theorem \ref{09_07_2020_thm_1}.

\begin{example}\label{09_10_2020_ex_2}
\begin{itemize}
    \item[(i)] Let $C=\la u(x-1)^6+u^2(x-1)(1+(x-1))+u^3\omega(x-1)^2 \ra$ be a cyclic code over $R=\mathbb F_4[u]/\la u^4 \ra$ of length $n=8$; here, $\mathbb F_4=\{0,1,\omega,\bar \omega\}$.
    Then, by Theorem \ref{09_07_2020_thm_1} (i), we calculate first
    $$u^3(x-1)^{\tau_2}\tilde h_2(x)=u^3(x-1)((1+(x-1))^2-\omega(x-1)^6)$$ 
    since $n-r_1+k_4=3<r_1=6$; hence, $\tau_2=1$ and $\tilde h_2(x)=(1+(x-1))^2-\omega(x-1)^6$ is a unit in $R[x]/\la x^n-1\ra$.
    The third torsional degree $t_3$ of $C$ is $\min\{n-r_1+k_4,\tau_2\}=\{3,1\}=1$. 
    
    \item[(ii)]
    Let $C=\la u(x-1)^6+u^2(x-1)(1+(x-1))+u^3\omega(x-1)^2, u^2(x-1)^4+u^3(1+\omega(x-1)+(x-1)^2) \ra$
    be a cyclic code over $R=\mathbb F_4[u]/\la u^4 \ra$ of length $n=8$.
    We note that
    $(k_4,k_5,k_6)=(1,2,0)$ and $(p_4(x),p_5(x),p_6(x))=(1+(x-1),\omega, 1+\omega(x-1)+(x-1)^2)$.
    By \eqref{09_09_2020_eq_5}, we have
    $$u^3(x-1)^{\tau_3}\tilde h_3(x)=u^3(\omega (x-1)^3-(1+\omega(x-1)+(x-1)^2)(1+(x-1)).$$
    By \eqref{09_09_2020_eq_1}, we get
    $$u^3(x-1)^{\tau_4}\tilde h_4(x)=u^3(x-1)(1+(x-1)-(x-1)(1+\omega(x-1)+(x-1)^2)).$$
    Hence, the third torsional degree $t_3$ of $C$ is 
    $$t_3=\min\{0,1,3,4,9\}=0$$
    since $\kappa=\min\{0,3\}=0$ and $t=1$; $t$ is the third torsional degree of $\la g_0(x)\ra$ obtained in (i).
\end{itemize}

\end{example}


For determining the third torsional degrees for the other cyclic codes over $\mathfrak R$, we use the following lemma.

\begin{lemma}\label{09_03_2020_lem_1}
Let $C$ be a cyclic code over $\mathfrak R$ of length $n$.
An arbitrary polynomial $g(x)$ in $C$ which has  $u^2$-part in $C$ is one of the following forms:
\begin{itemize}

\item[{\em (i)}] If $C=\langle g_0(x)\rangle$, then the polynomial $g(x)$ is an element of the set
 $T_1=\{\eqref{09_03_2020_eq_2},\eqref{09_07_2020_eq_13}, \eqref{09_06_2020_eq_1}\}$.

\item[{\em (ii)}] If $C=\langle g_0(x), g_1(x) \rangle$, then the polynomial $g(x)$ is an element of the set 
 $T_2=\{\eqref{09_03_2020_eq_2}, \eqref{09_03_2020_eq_3}, \eqref{09_02_2020_eq_1},
 \eqref{09_07_2020_eq_12}, \eqref{09_07_2020_eq_13}, \eqref{09_06_2020_eq_2}, \eqref{09_06_2020_eq_1}, \eqref{09_06_2020_eq_3}\}$.

\item[{\em (iii)}] If $C=\langle g_0(x), g_1(x), g_2(x)\rangle$, then the polynomial $g(x)$ is an element of the set
$T_3=\{\eqref{09_03_2020_eq_2},   \eqref{09_03_2020_eq_3}, \eqref{09_02_2020_eq_1},
\eqref{09_07_2020_eq_12},\eqref{09_07_2020_eq_13},  \eqref{09_06_2020_eq_2}, \eqref{09_06_2020_eq_1}, \eqref{09_06_2020_eq_3}, \eqref{09_07_2020_eq_11}\}$.

\item[{\em (iv)}] If $C=\langle g_0(x),g_2(x)\rangle$, then the polynomial $g(x)$ is an element of the set
$T_4=\{\eqref{09_03_2020_eq_2},\eqref{09_07_2020_eq_13}, \eqref{09_06_2020_eq_1}, \eqref{09_07_2020_eq_11} \}$,
\end{itemize}

where

\begin{eqnarray}
u^2(x-1)^{\epsilon_1}h_1(x)+u^3(x-1)^{\epsilon_2}p_3(x)-u^3(x-1)^{\epsilon_3}p_1(x)p_2(x),\label{09_03_2020_eq_2}\\
u^2(x-1)^{\epsilon_4}h_2(x)+u^3(x-1)^{\epsilon_5}p_3(x)-u^3(x-1)^{\epsilon_6}p_1(x)p_5(x),\label{09_03_2020_eq_3}\\
u^2(x-1)^{min\{k_1,r-r_1+k_4\}+\hat e_3}h_3(x) + u^3(x-1)^{k_2}p_2(x)-u^3(x-1)^{r-r_1+k_5}p_5(x),\label{09_02_2020_eq_1}\\
u^2(x-1)^{n-k_1+k_2}p_2(x)+u^3(x-1)^{n-r_1+k_3}p_3(x),\label{09_07_2020_eq_12}\\
u^2(x-1)^{n-r+k_1}p_1(x)+u^3(x-1)^{n-r+k_2}p_2(x),\label{09_07_2020_eq_13}  \\  
u^2(x-1)^{n-r_1+k_4}p_4(x)+u^3(x-1)^{n-r_1+k_5}p_5(x), \label{09_06_2020_eq_2}\\
u^2(x-1)^r+u^3(x-1)^{k_1}p_1(x), \label{09_06_2020_eq_1} \\
u^2(x-1)^{r_1}+u^3(x-1)^{k_4}p_4(x),\label{09_06_2020_eq_3}\\
u^2(x-1)^{r_2}+u^3(x-1)^{k_6}p_6(x),\label{09_07_2020_eq_11}
\end{eqnarray}
with
{\small \[
\begin{array}{l}
(\epsilon_1,\epsilon_2,\epsilon_3)\\
=\begin{cases}
(\min\{n-r+k_2,n-2r+2k_1\}+\hat e_1, n-r+k_3,  & \mbox{if~} n-r+k_1>r,\\
~~~~~~~n-2r+k_1+k_2)&\\
(\min\{k_1, r-k_1+k_2\}+\hat e_1, r-k_1+k_3, k_2) & \mbox{if~} n-r+k_1\le r,
\end{cases}
\end{array}
\]}
{\small \[
\begin{array}{l}
(\epsilon_4,\epsilon_5,\epsilon_6)\\
=\begin{cases}
(\min\{n-r+k_2,n-r+k_1-r_1+k_4\}+\hat e_2,  & \mbox{if~} n-r+k_1>r_1,\\
~~~~~~~n-r+k_3, n-r+k_1-r_1+k_5)&\\
(\min\{k_4,r_1-k_1+k_2\}+\hat e_2, r_1-k_1+k_3, k_5) & \mbox{if~} n-r+k_1\le r_1,
\end{cases}
\end{array}
\]}
{\small \[
\begin{array}{l}
(x-1)^{\epsilon_1}h_1(x)\\
=\begin{cases}
(x-1)^{n-r+k_2}p_2(x)-(x-1)^{n-2r+2k_1}p_1^2(x) & \mbox{if~} n-r+k_1>r,\\
(x-1)^{r-k_1+k_2}p_2(x) -(x-1)^{k_1}p_1^2(x)& \mbox{if~} n-r+k_1\le r,
\end{cases}
\end{array}
\]}
{\small \[
\begin{array}{l}
(x-1)^{\epsilon_4}h_2(x)\\
=\begin{cases}
(x-1)^{n-r+k_2}p_2(x)-(x-1)^{n-r+k_1-r_1+k_4}p_1(x)p_4(x) & \mbox{if~} n-r+k_1>r_1,\\
(x-1)^{r_1-k_1+k_2}p_2(x)-(x-1)^{k_4}p_1(x)p_4(x) & \mbox{if~} n-r+k_1\le r_1,
\end{cases}
\end{array}
\]}
$(x-1)^{min\{k_1,r-r_1+k_4\}+\hat e_3}h_3(x)=(x-1)^{k_1}p_1(x)-(x-1)^{r-r_1+k_4}p_4(x)$,\\
$\hat e_i$ is a non-negative integer, and $h_i(x)$ is a unit or zero in $\FR$ for $i=1,2,3$.

\end{lemma}

\begin{proof}
Suppose that $C=\langle g_0(x), g_1(x)\rangle$.
Then any element in $C$ is written as
\begin{equation}\label{09_07_2020_eq_1}
\begin{array}{l}
g_0(x)F_1(x)+g_1(x)F_2(x)\\
=((x-1)^r+u(x-1)^{k_1}p_1(x)+u^2(x-1)^{k_2}p_2(x)+u^3(x-1)^{k_3}p_3(x))F_1(x)\\
+(u(x-1)^{r_1}+u^2(x-1)^{k_4}p_4(x)+u^3(x-1)^{k_5}p_5(x))F_2(x),\\
\end{array}
\end{equation}

where $F_1(x)=(f_{00}(x)+uf_{01}(x)+u^2f_{02}(x)+u^3f_{03}(x))$, $F_2(x)=(f_{10}(x)+uf_{11}(x)+u^2f_{12}(x))$ and $f_{ij}(x)\in \FR$ for $i=0,1$ and $j=0,1,2,3$.
If an element $g(x)$ in $C$ is divisible by $u$, then $f_{00}(x)$ must be equal to $(x-1)^{n-r}\tilde f_{00}(x)$ in \eqref{09_07_2020_eq_1} for some $\tilde f_{00}(x) \in \FR$.

Then the polynomial $g(x)$ which has nonzero $u$-part in $C$ is one of the following form:
{\small
\begin{eqnarray}
(u(x-1)^{n-r+k_1}p_1(x)+u^2(x-1)^{n-r+k_2}p_2(x)+u^3(x-1)^{n-r+k_3}p_3(x))\tilde f_{00}(x),\label{09_07_2020_eq_2}\\
(u(x-1)^{r}+u^2(x-1)^{k_1}p_1(x)+u^3(x-1)^{k_2}p_2(x))f_{01}(x),\label{09_07_2020_eq_3}\\
(u(x-1)^{r_1}+u^2(x-1)^{k_4}p_4(x)+u^3(x-1)^{k_5}p_5(x))f_{10}(x)\label{09_07_2020_eq_4}
\end{eqnarray}}
from \eqref{09_07_2020_eq_1}.
We need to annihilate $u$-part in each polynomials \eqref{09_07_2020_eq_2}, \eqref{09_07_2020_eq_3} and \eqref{09_07_2020_eq_4}.
First, we calculate the following using \eqref{09_07_2020_eq_3} and \eqref{09_07_2020_eq_4} since $r>r_1$:
\[
\begin{array}{l}
u(x-1)^{r}+u^2(x-1)^{k_1}p_1(x)+u^3(x-1)^{k_2}p_2(x)\\
-(x-1)^{r-r_1}(u(x-1)^{r_1}+u^2(x-1)^{k_4}p_4(x)+u^3(x-1)^{k_5}p_5(x))\\
=u^2(x-1)^{k_1}p_1(x)-u^2(x-1)^{r-r_1+k_4}p_4(x)\\
+u^3(x-1)^{k_2}p_2(x)-u^3(x-1)^{r-r_1+k_5}p_5(x),\\
=u^2(x-1)^{\min\{k_1, r-r_1+k_4\}+\hat e_3}h_3(x)
+u^3(x-1)^{k_2}p_2(x)-u^3(x-1)^{r-r_1+k_5}p_5(x),
\end{array}
\]
where $(x-1)^{min\{k_1,r-r_1+k_4\}+\hat e_3}h_3(x)=(x-1)^{k_1}p_1(x)-(x-1)^{r-r_1+k_4}p_4(x)$;
this calculating result is equation \eqref{09_02_2020_eq_1}.

\begin{itemize}
\item[(a)] Assume $n-r+k_1>r_1$.
For the same reasoning as above, we have the followings using \eqref{09_07_2020_eq_2} and \eqref{09_07_2020_eq_4}: 
{\small \[
\begin{array}{l}
u(x-1)^{n-r+k_1}p_1(x)+u^2(x-1)^{n-r+k_2}p_2(x)+u^3(x-1)^{n-r+k_3}p_3(x)\\
-(x-1)^{n-r+k_1-r_1}p_1(x)(u(x-1)^{r_1}+u^2(x-1)^{k_4}p_4(x)+u^3(x-1)^{k_5}p_5(x))\\
=u^2(x-1)^{n-r+k_2}p_2(x)-u^2(x-1)^{n-r+k_1-r_1+k_4}p_1(x)p_4(x)\\
+u^3(x-1)^{n-r+k_3}p_3(x)-u^3(x-1)^{n-r+k_1-r_1+k_5}p_1(x)p_5(x),\\
=u^2(x-1)^{\min\{n-r+k_2,n-r+k_1-r_1+k_4\}+\hat e_2}h_2(x)\\
+u^3(x-1)^{n-r+k_3}p_3(x)-u^3(x-1)^{n-r+k_1-r_1+k_5}p_1(x)p_5(x),
\end{array}
\]}
where $h_2(x)$ is a unit or zero in $\FR$ such that $(x-1)^{n-r+k_2}p_2(x)-(x-1)^{n-r+k_1-r_1+k_4}p_1(x)p_4(x)=(x-1)^{\min\{n-r+k_2,n-r+k_1-r_1+k_4\}+\hat e_2}h_2(x)$.

\item[(b)] Assume $n-r+k_1\le r_1$.
We get the followings using \eqref{09_07_2020_eq_2} and \eqref{09_07_2020_eq_4}: 
{\small \[
\begin{array}{l}
(x-1)^{r_1-n+r-k_1}(u(x-1)^{n-r+k_1}p_1(x)+u^2(x-1)^{n-r+k_2}p_2(x)\\
+u^3(x-1)^{n-r+k_3}p_3(x))-p_1(x)(u(x-1)^{r_1}+u^2(x-1)^{k_4}p_4(x)+u^3(x-1)^{k_5}p_5(x))\\
=u^2(x-1)^{r_1-k_1+k_2}p_2(x)-u^2(x-1)^{k_4}p_1(x)p_4(x)
+u^3(x-1)^{r_1-k_1+k_3}p_3(x)\\
-u^3(x-1)^{k_5}p_1(x)p_5(x),\\
=u^2(x-1)^{\min\{r_1-k_1+k_2,k_4\}+\hat e_2}h_2(x)
+u^3(x-1)^{r_1-k_1+k_3}p_3(x)\\
-u^3(x-1)^{k_5}p_1(x)p_5(x),
\end{array}
\]}
where $h_2(x)$ is a unit or zero in $\FR$ such that $(x-1)^{r_1-k_1+k_2}p_2(x)-(x-1)^{k_4}p_1(x)p_4(x)=(x-1)^{\min\{r_1-k_1+k_2,k_4\}+\hat e_2}h_2(x)$.
\end{itemize}

From (a) and (b), we obtain the equation \eqref{09_03_2020_eq_2}.

By the same process as (a) and (b), according to $n-r+k_1>r$ or $n-r+k_1\le r$, we obtain the equation \eqref{09_03_2020_eq_3} using equations  \eqref{09_07_2020_eq_2} and \eqref{09_07_2020_eq_3}.

Moreover, from equation \eqref{09_07_2020_eq_1}, we obtain all polynomials such that $u$-part is zero with nonzero $u^2$-part:
\[
\begin{array}{l}
\mbox{equation~} \eqref{09_07_2020_eq_12}=g_0(x)(x-1)^{n-k_1}\hat f_{00}(x),\\
\mbox{equation~}\eqref{09_07_2020_eq_13}=g_0(x)(x-1)^{n-r}(u\tilde f_{01}(x)),\\
\mbox{equation~}\eqref{09_06_2020_eq_2}=g_1(x)(x-1)^{n-r_1}\tilde f_{10}(x),\\
\mbox{equation~}\eqref{09_06_2020_eq_1}=g_0(x)(u^2f_{02}(x)),\\
\mbox{equation~}\eqref{09_06_2020_eq_3}=g_1(x)(u f_{11}(x)),
\end{array}
\]
where $f_{00}(x)=(x-1)^{n-r}\tilde f_{00}(x)=(x-1)^{n-r}(x-1)^{r-k_1}\hat f_{00}(x)=(x-1)^{n-k_1}\hat f_{00}(x)$,
$f_{01}(x)=(x-1)^{n-r}\tilde f_{01}(x)$
and
$f_{10}(x)=(x-1)^{n-r_1}\tilde f_{10}(x)$ 
for some polynomials $\tilde f_{00}
(x)$, $\hat f_{00}(x)$, $\tilde f_{01}(x)$ and $\tilde f_{10}(x)$ in $\R$.
Hence, for $C=\langle g_0(x), g_1(x) \rangle$, any polynomial which has nonzero $u^2$-part in $C$ is an element of the set $T_2$.

The other parts can be proved similar way as above.
\end{proof}

We note that every polynomial in Lemma \ref{09_03_2020_lem_1} has the following form:
\begin{equation}\label{09_02_2020_eq_6}
u^2(x-1)^{k}h_1(x)+u^3(x-1)^{\tilde k}h_2(x),
\end{equation}
where $k$ and $\tilde k$ are non-negative integers, and $h_i(x)$ is a unit or zero in $\FR$ for $i=1,2$.
In \eqref{09_02_2020_eq_6}, the degree $k$ is (resp. $\tilde k$) is called a {\it $u^2$-degree} (resp. {\it $u^3$-degree}) when $h_1(x)\neq 0$ (resp. $h_2(x)\neq 0$).

\bigskip
In Theorem \ref{09_07_2020_thm_2}, we determine the third torsional degree for cyclic codes $\la g_0(x)\ra$, $\la g_0(x),g_1(x)\ra$, $\la g_0(x), g_2(x) \ra$ and $\la g_0(x), g_1(x), g_2(x) \ra$ in $\R$.

\begin{theorem}\label{09_07_2020_thm_2}
Let C be a cyclic code over $\mathfrak R$ of length $n$, and $T$ be the set defined as in Lemma \ref{09_03_2020_lem_1}:
\[
T=\begin{cases}
T_1 & \mbox{if~} C=\la g_0(x) \ra,\\
T_2 & \mbox{if~} C=\la g_0(x), g_1(x)\ra,\\
T_3 & \mbox{if~} C=\la g_0(x), g_1(x), g_2(x)\ra,\\
T_4 & \mbox{if~} C=\la g_0(x), g_2(x)\ra.\
\end{cases}
\]
Let $\nu$ be the number of polynomials in $T_i$ $(1 \le i \le 4)$ which has a non-negative $u^2$-degree for each cases.
Let $f_i=u^2(x-1)^{\omega_i}h_{i1}(x)+u^3(x-1)^{\tilde \omega_i}h_{i2}(x)$ be an element of the set $T$, where $\omega_i$ (resp. $\tilde \omega_i$) is the $u^2$-degree (resp. $u^3$-degree) and $h_{ij}(x)$ is a unit or zero in $\FR$ $(1\le i \le |T|, j=1,2)$.
Then the third torsional degree $t_3$ of the code $C$ is obtained as follows:
\begin{itemize}
\item[{\em (i)}] If $\nu=0$, then the third torsional degree $t_3=\min\{\tilde \omega_i : 1 \le i \le |T|\}$.

\item[{\em (ii)}] If $\nu=1$, we suppose that $f_1(x)$ has a non-negative $u^2$-degree.
Then the third torsional degree $t_3$ is 
$$t_3=\min\{n-\omega_1+\tilde \omega_1, \omega_1, \tilde \omega_2,\ldots,\tilde \omega_{|T|}\}.$$

\item[{\em (iii)}] Suppose that $\nu\ge 2$. Set $f_1(x)$, \ldots, $f_{\nu}(x)$ have non-negative $u^2$-degrees with $\omega_1\le \ldots \le \omega_\nu$.
We have that 
\begin{equation}\label{09_10_2020_eq_1}
f_j-(x-1)^{\omega_j-\omega_i}f_i(x)h_{i1}^{-1}(x)h_{j1}(x)=u^3(x-1)^{\tau_k}\tilde h_k(x),
\end{equation}
where $\tilde h_k(x)$ is a unit or zero in $\R$ $(1\le i<j \le \nu, 1\le k \le \binom{\nu}{2})$.
Let $M=\{\tau_k: \tilde h_k(x)\neq 0\}$, and $\mathfrak m$ be the minimum value in $M$.
Then the third torsional degree $t_3$ is 
{\small \[
t_3=\begin{cases}
\min\{\mathfrak m, \omega_1, \ldots,\omega_\nu, n-\omega_1+\tilde \omega_1,\ldots,n-\omega_\nu+\tilde \omega_\nu,
 \tilde \omega_{\nu+1},\ldots, \tilde \omega_{|T|}\} & \mbox{if~}M\neq \emptyset,\\
\min\{\omega_1, \ldots,\omega_\nu, n-\omega_1+\tilde \omega_1,\ldots,n-\omega_\nu+\tilde \omega_\nu,
 \tilde \omega_{\nu+1},\ldots, \tilde \omega_{|T|}\} & \mbox{if~} M = \emptyset.
\end{cases}
\]}

\end{itemize}
\end{theorem}

\begin{proof}
\begin{itemize}
    \item[{(i)}] Suppose that $\nu=0$.
    It means that every polynomial in the set $T$ has zero $u^2$-degree.
    Hence we only consider $u^3$-degrees $\tilde \omega_i$ $(1\le i \le |T|)$ for all polynomials in $T$.
    Easily, the third torsional degree for the code is
    $$\min\{\tilde \omega_i : 1 \le i \le |T|\}.$$
    \item[{(ii)}] Suppose that $\nu=1$.
    Then for removing $u^2(x-1)^{\omega_1}\tilde h_{11}(x)$ for $f_1(x)$, we multiply $(x-1)^{n-\omega_1}$ to $f_1(x)$.
    We obtain that $(x-1)^{n-\omega_1}f_1(x)=u^3(x-1)^{n-\omega_1+\tilde \omega_1}h_{12}(x)$.
    Furthermore, $uf_1(x)=u^3(x-1)^{\omega_1}h_{11}(x)$ is also an element in $C$.
    These imply the reulsts.
    \item[{(iii)}] We assume that $\nu\ge 2$.
    Then for the polynomials $f_1,\ldots,f_{\nu}$, we obtain the same result as the previous case (ii).
    On the other way, for deleting $u^2$-parts of $f_i(x)$ and $f_j(x)$ each other, we compute
    $$f_j(x)-(x-1)^{\omega_j-\omega_i}f_i(x)h_{i1}^{-1}(x)h_{j1}(x)=u^3(x-1)^{\tau_k}\tilde h_k(x),$$
    where $1 \le i < j \le \nu$ and $1 \le k \le \binom{\nu}{2}$.
    Through this process, we have that $u^3$-degree $\tau_k$ only when $\tilde h_k(x)\neq 0$.
    Hence we get the result.
\end{itemize}
The items (i), (ii) and (iii) therefore implies the third torsional degree $t_3$ for the code $C$ over $\mathfrak R$.
\end{proof}

Finally, in Theorem \ref{09_09_2020_thm_1}, we get the third torsional degree of a cyclic code which is a non-principal ideal containing $g_3(x)=u^3(x-1)^{r_3}$ using Theorems \ref{09_07_2020_thm_1} and \ref{09_07_2020_thm_2}.

\begin{theorem}\label{09_09_2020_thm_1}
Let $\hat C$ be a cyclic code given in Theorems \ref{09_07_2020_thm_1} and \ref{09_07_2020_thm_2} except for the case $\hat C=\la g_3(x)\ra$.
Let $\mathfrak t$ be the third torsional degree for a cyclic code $\hat C$.
Let $C$ be a cyclic code over $\mathfrak R$ of length $n$ generated by $g_3(x)$ and all generator polynomials in $\hat C$.
Then the third torsional degree $t_3$ for the code $C$ is obtained as follows:
\begin{itemize}
    \item[{\em (i)}] If $r_3\ge \mathfrak t$, then the third torsional degree $t_3$ of $C$ is equal to $\mathfrak t$.
    \item[{\em (ii)}] If $r_3<\mathfrak t$, then the third torsional degree $t_3$ of $C$ is equal to $r_3$.
\end{itemize}
\end{theorem}

\begin{proof}
If $r_3\ge \mathfrak t$, then $C=\hat C$, 
that is $g_3(x)\in \hat C$ 
since $u^3(x-1)^{\mathfrak t}\in \hat C$ and $u^3(x-1)^{\mathfrak t}(x-1)^{r_3-\mathfrak t}=u^3(x-1)^{r_3}$.
Hence, in this case, $t_3=\mathfrak t$.

If $r_2<\mathfrak t$, then $r_3$ is the smallest non-negative integer such that $u^3(x-1)^{r_2}\in C$.
Thus, it follows the result.
\end{proof}

The next example gives the third torsional degree for certain cyclic codes by using Theorem \ref{09_07_2020_thm_2}.

\begin{example}\label{09_10_2020_ex_3}
\begin{itemize}
    \item[(i)] Let $C=\la (x-1)^3+u(x-1)+u^2 + u^3(1+(x-1)), u(x-1)^2+u^2+u^3(x-1)\ra$ be a cyclic code over $\Bbb F_2[u]/\la u^4 \ra$ of length $4$.
    We note that $(k_1,k_2,k_3,k_4,k_5)=(1,0,0,0,1)$ and $(p_1(x),p_2(x),p_3(x),p_4(x),p_5(x))=(1,1,1+(x-1),1,1)$.
    We have the followings
    \[
    \begin{array}{l}
    \bullet~ (x-1)^{\epsilon_1}h_1(x)=(x-1)(1+(x-1)),\\
    \bullet~ (x-1)^{\epsilon_4}h_2(x)=-1+(x-1),\\
    \bullet~ (x-1)^{\min\{k_1,r-r_1+k_4\}+\hat e_3}h_3(x)=0,\\
    \bullet~ (\epsilon_2,\epsilon_3)=(r-k_1+k_3,k_2)=(2,0),\\
    \bullet~ (\epsilon_5,\epsilon_6)=(r_1-k_1+k_3,k_5)=(1,1),
    
    \end{array}
    \]
    since $n-r+k_1=2<r$ and $n-r+k_1\le r_1$.
    Thus, by Lemma \ref{09_03_2020_lem_1}, the set $T_3$ consists of the following polynomials:
    \[
    \begin{array}{l}
    u^2(x-1)(1+(x-1))+u^3(x-1)^2(1+(x-1))-u^3,\\
    u^2(-1+(x-1))+u^3(x-1)(1+(x-1))-u^3(x-1),\\
    u^3-u^3(x-1)^2,\\
    u^2(x-1)^3+u^3(x-1)^2(1+(x-1)),\\
    u^2(x-1)^2+u^3(x-1),\\
    u^2(x-1)^2+u^3(x-1)^3,\\
    u^2(x-1)^3+u^3(x-1),\\
    u^2(x-1)^2+u^3.
    \end{array}
    \]
    In Theorem \ref{09_07_2020_thm_2}, we get $t_3=0$ since there is the polynomial $u^3-u^3(x-1)^2=u^3(1-(x-1)^2)$.
    It means that $u^3\in C$.
    
    \item[(ii)] Let $C=\la (x-1)^5+u(x-1)^2(1+2(x-1))+u^3(x-1)(3+(x-1))\ra$ be a cyclic code over $\mathbb F_3[u]/\la u^4 \ra$ of length $9$.
    We note that $(k_1,k_2,k_3)=(2,0,1)$ and $(p_1(x),p_2(x),p_3(x))=(1+2(x-1),0,3+(x-1))$.
    We have that
    \[
    \begin{array}{l}
    \bullet~ (x-1)^{\epsilon_1}h_1(x)=0,\\
    \bullet~ (\epsilon_2,\epsilon_3)=(n-r+k_3=5, n-2r+k_1+k_2=1),\\
    \end{array}\]
    since $n-r+k_1=6>r$.
    Then the set $T_1$ consists of 
    \[
    \begin{array}{l}
    u^3(x-1)^5(1+2(x-1))(3+(x-1)),\\
    u^2(x-1)^6(1+2(x-1)),\\
    u^2(x-1)^5+u^3(x-1)^2(1+2(x-1)).
    \end{array}\]
    Hence $\nu=2$ here.
    Let $f_1(x)=u^2(x-1)^6(1+2(x-1))$ and $f_2(x)=u^2(x-1)^5+u^3(x-1)^2(1+2(x-1))$.
    Then $f_1(x)-(x-1)(1+2(x-1))f_2(x)=-u^3(x-1)^3(1+2(x-1))^2$, that is, $\tau_1=3$ in equation \eqref{09_10_2020_eq_1}.
    Thus, the third torsional degree $t_3$ of $C$ is $\min\{3,5,6\}=3$ by Theorem \ref{09_07_2020_thm_2}.
\end{itemize}
\end{example}

By \cite[Theorems 18 and 22]{SS}, the minimum symbol-pair weight and RT-weight for a cyclic code $C$ over $\mathfrak R$ can be given by the third torsional degree $t_3$ of $C$.

\begin{theorem}\label{09_10_2020_cor_1}
Let $C$ be a cyclic code over $\mathfrak R=\mathbb F_{p^m}[u]/\la u^4 \ra$ of length $n=p^k$.
Let $Tor_3(C)=\la (x-1)^{t_3}\ra$ be the third torsion code of $C$ over $\Bbb F_{p^m}$ with the third torsional degree $t_3$ 
(the third torsional degree $t_3$ is obtained in Theorems \ref{09_07_2020_thm_1}, \ref{09_07_2020_thm_2} and \ref{09_09_2020_thm_1}).
\begin{itemize}
    \item[{\em (i)}] The minimum symbol-pair weight $wt_{sp}(C)$ of $C$ is 
    \[
    \begin{cases}
    2 & \mbox{if~} t_3=0,\\
    3p^\ell & \mbox{if~} t_3=p^k-p^{k-\ell}+1, \mbox{~where~}0\le \ell\le k-2,\\
    4p^\ell & \mbox{if~} p^k-p^{k-\ell}+2\le t_3\le p^k-p^{k-\ell}+p^{k-\ell-1},\\
    &\mbox{where~}0\le \ell\le k-2,\\
    2(\mu+2)p^\ell & \mbox{if~} p^k-p^{k-\ell}+\mu p^{k-\ell-1}+1 \le t_3\le p^k-p^{k-\ell}+(\mu+1)p^{k-\ell-1},\\
    &\mbox{where~}0\le \ell\le k-2 \mbox{~and~} 1\le \mu \le p-2,\\
    (\mu+2)p^{k-1} & \mbox{if~} t_3=p^k-p+\mu, \mbox{~where~} 1\le \mu \le p-2,\\
    p^k & \mbox{if~} t_3=p^k-1,\\
    0 & \mbox{if~} t_3=p^k.
    \end{cases}
    \]
    \item[{\em (ii)}] The RT weight $wt_{RT}(C)$ of $C$ is $t_3+1$.
\end{itemize}
\end{theorem}

We close this section with the next example; we obtain symbol-pair weight and RT weight for certain cyclic codes.

\begin{example}\label{09_17_2020_ex_1}
\begin{itemize}
    \item[(i)] In Examples \ref{09_10_2020_ex_2} and \ref{09_10_2020_ex_3}, we obtain symbol-pair weight and RT weight for certain cyclic codes using Theorem \ref{09_10_2020_cor_1}:
    \[
    (wt_{sp}(C), wt_{RT}(C))=\begin{cases}
    (3,2) & \mbox{$C$ is given in Example \ref{09_10_2020_ex_2} (i)},\\
    (2,1) & \mbox{$C$ is given in Examples \ref{09_10_2020_ex_2} (ii) and \ref{09_10_2020_ex_3} (i)}, \\
    (4,4) & \mbox{$C$ is given in Example \ref{09_10_2020_ex_3} (ii)}.
    \end{cases}
    \]
    
    \item[(ii)] Let $C=\la u^2(x-1)^{51}+u^3(x-1)^{67}h(x) \ra$ be a cyclic code over $R=\mathbb F_{5^2}[u]/\la u^4 \ra$ of length $n=5^3$, where $h(x)$ is an arbitrary unit in $R[x]/\la x^n-1\ra$.
    Then, by Theorem \ref{09_07_2020_thm_1}, the third torsional degree $t_3$ of $C$ is $t_3=\min\{141, 51\}=51$ since $n-r_2+k_6=141$ and $r_2=51$.
    From Theorem \ref{09_10_2020_cor_1}, we have $wt_{sp}(C)=8$ and $wt_{RT}(C)=52$.
    
    \item[(iii)] Generally, let $C=\la u^2(x-1)^{r_2}+u^3(x-1)^{k_6}p_6(x)\ra$ be a cyclic code over $\mathbb F_{5^2}[u]/\la u^4 \ra$ of length $n=5^3$.
    For each cases, we obtain the symbol-pair weight and RT weight for the code $C$ from Theorem \ref{09_10_2020_cor_1}.\\
    (a) If $0\le r_2 \le 62$, then the third torsinoal degree $t_3$ is equal to $r_2$ by Theorem \ref{09_07_2020_thm_1}.
    Thus $wt_{RT}(C)=r_2+1$, and
    \begin{equation}\label{09_17_2020_eq_1}
    wt_{sp}(C)=\begin{cases}
    2 & \mbox{if~} t_3=0,\\
    3 & \mbox{if~} t_3=1,\\
    4 & \mbox{if~} 2 \le t_3 \le 25,\\
    6 & \mbox{if~} 26 \le t_3 \le 50,\\
    8 & \mbox{if~} 51 \le t_3 \le 62.\\
    \end{cases}
    \end{equation}
    
    (b) If $63\le r_2 \le 124$ with $2r_2-k_6\le 125$; so the third torsional degree $t_3$ is $r_2$.
    Then $wt_{RT}(C)=r_2+1$, and
    \begin{equation}\label{09_17_2020_eq_2}
    wt_{sp}(C)=\begin{cases}
    8 & \mbox{if~} 63 \le t_3 \le 75,\\
    10 & \mbox{if~} 76 \le t_3 \le 100,\\
    15 & \mbox{if~} t_3=101,\\
    20 & \mbox{if~} 102 \le t_3 \le 105,\\
    30 & \mbox{if~} 106 \le t_3 \le 110,\\
    40 & \mbox{if~} 111 \le t_3 \le 115,\\
    50 & \mbox{if~} 116 \le t_3 \le 120,\\
    75 & \mbox{if~} t_3=121,\\
    100 & \mbox{if~} t_3=122,\\
    125 & \mbox{if~} t_3=123,\\
    125 & \mbox{if~} t_3=124,\\
    0 & \mbox{if~} t_3=125.\\
    \end{cases}
    \end{equation}
    
    (c) If $63\le r_2 \le 124$ with $2r_2-k_6<125$, then $t_3=n-r_2+k_6$.
    Hence, $wt_{RT}(C)=n-r_2+k_6+1$.
    By \eqref{09_17_2020_eq_1} and \eqref{09_17_2020_eq_2}, according to the value $t_3=n-r_2+k_6$, we can obtain $wt_{sp}(C)$.
\end{itemize}
\end{example}

\end{document}